\documentclass[conference]{IEEEtran}
\IEEEoverridecommandlockouts
\usepackage{cite}
\usepackage{amsmath,amssymb,amsfonts}
\usepackage{graphicx}
\usepackage{textcomp}
\usepackage{xcolor}
\def\BibTeX{{\rm B\kern-.05em{\sc i\kern-.025em b}\kern-.08em
    T\kern-.1667em\lower.7ex\hbox{E}\kern-.125emX}}
\usepackage{graphicx} % Required for inserting images
\usepackage{subfig}
\usepackage{comment}
\newcommand{\F}{\mathbb{F}}
\newcommand{\X}{\boldsymbol{X}}

\newcommand{\code}{\mathcal{C}}
\newcommand{\Y}{\boldsymbol{Y}}
\newcommand{\x}{\boldsymbol{x}}
\newcommand{\y}{\boldsymbol{y}}
\newcommand{\w}{\boldsymbol{w}}
\newcommand{\bfu}{\boldsymbol{u}}
\newcommand{\bfv}{\boldsymbol{v}}
\newcommand{\bfa}{\boldsymbol{a}}
\newcommand{\bfb}{\boldsymbol{b}}
\newcommand{\bfe}{\boldsymbol{e}}
\newcommand{\bfz}{\boldsymbol{z}}
\newcommand{\z}{\boldsymbol{z}}
\newcommand{\Z}{\boldsymbol{Z}}
\newcommand{\ham}{d_{\text{H}}}
\usepackage{verbatim}
\usepackage[font=small]{caption}
\usepackage{algorithm}
\usepackage{relsize}
\usepackage{amsthm}
\usepackage{algpseudocode}
\usepackage{mathrsfs}
\usepackage{dsfont}
\usepackage{relsize}
\usepackage{enumitem}
\usepackage{hyperref}
\usepackage{xcolor}
\newtheorem{theorem}{Theorem}

\newtheorem{lemma}{Lemma}
\newtheorem{definition}{Definition}
\newtheorem{example}{Example}
\newtheorem{claim}{Claim}

\newtheorem{corollary}{Corollary}

\usepackage{tikz}
\usetikzlibrary{shapes.geometric, arrows}
\theoremstyle{definition}
\newtheorem{construction}{Construction}
\newtheorem{alg}{Algorithm}

\newcommand{\ey}[1]{{
[{\textcolor{red}{#1}} \textcolor{red}{--eitan}]\normalsize}}

\begin{document}

\title{Error-Correcting Codes for the Sum Channel}

\author{\IEEEauthorblockN{Lyan Abboud}
\IEEEauthorblockA{\textit{Computer Science Department,} \\
\textit{Technion—Israel Institute of Technology}\\
Haifa, Israel \\
lyan.abboud@campus.technion.ac.il\vspace{-4ex}}
\and
\IEEEauthorblockN{Eitan Yaakobi}
\IEEEauthorblockA{\textit{Computer Science Department,} \\
\textit{Technion—Israel Institute of Technology}\\
Haifa, Israel \\
yaakobi@cs.technion.ac.il\vspace{-4ex}}

}

\maketitle

\begin{abstract}
We introduce the \emph{sum channel}, a new channel model motivated by applications in distributed storage and DNA data storage. In the error-free case, it takes as input an $\ell$-row binary matrix and outputs an $(\ell+1)$-row matrix whose first $\ell$ rows equal the input and whose last row is their parity (sum) row.
We construct a two-deletion-correcting code with redundancy $2\lceil\log_2\log_2 n\rceil + \log_2 \ell+ O(1)$ for $\ell$-row inputs. When $\ell=2$, we establish a lower bound of $\lceil\log_2\log_2 n\rceil + O(1)$ bits, implying that our redundancy is optimal up to a factor of~2.
 We also present a code correcting a single substitution with $\lceil \log_2(\ell+1)\rceil$ redundant bits and prove that it is within one bit of optimality.
\end{abstract}

\section{Introduction}
In this work, we study the \emph{sum channel}, which takes as input an $\ell$-row binary matrix and, in the error-free case, outputs an $(\ell+1)$-row matrix whose first $\ell$ rows equal the input and whose last row is their parity row, thereby introducing inherent redundancy. We analyze this channel under output errors including bit insertions, deletions, and substitutions, and derive coding schemes and upper bounds on the maximum code size. 

One natural application of the sum channel arises in distributed storage and matches its structure. 
In Redundant Array of Independent Disks Level 5 (RAID~5) \cite{raid,chen1994raid}, a widely used storage layout, data are striped across $\ell$ disks and an additional \emph{parity} stripe is stored as the bitwise XOR of the data stripes; this extra stripe enables recovery from the loss or corruption of a single disk, while practical systems must also cope with bit-level errors and mis-synchronization. 
%To the best of our knowledge, using the parity disk to correct local errors has not been studied.
To the best of our knowledge, local error correction for the RAID~5 single-parity sum structure has not been treated in the coding theory literature.

Another application arises in DNA sequencing and data storage. Storing digital information on DNA involves encoding it into sequences over the DNA alphabet $\{A,C,G,T\}$, synthesizing DNA molecules with the desired sequences, and storing the resulting material. To retrieve the data, the strands undergo \emph{DNA sequencing}, which determines the order of the symbols along the strand, often by incorporating nucleotides one at a time and identifying them through fluorescent signals. The resulting reads are then fed to a decoder to reconstruct the stored information. 
DNA synthesis and sequencing are error-prone, introducing mostly substitutions, insertions, and deletions; therefore, error-correcting codes are often used to ensure reliability~\cite{blawat2016forward}\cite{cai2019optimal}.
An innovative DNA sequencing strategy, termed \emph{ECC sequencing}~\cite{chen2017highly}, suggests reading the same DNA sequence three times using three distinct two-class partitions of $\{A,C,G,T\}$. Each round outputs a binary vector indicating, for each position, the class of the underlying base. The three distinct partitions are:
$\{A,C\}\ \text{vs.}\ \{G,T\},\ 
\{A,G\}\ \text{vs.}\ \{C,T\}$ and $ \{A,T\}\ \text{vs.}\ \{C,G\}$.
The binary vector corresponding to a partition records a bit $0$ if the true base lies in the first class (and $1$ otherwise). 
\begin{figure}[h]
    \vspace{-7pt}
    \centering
    \includegraphics[width=1\linewidth]{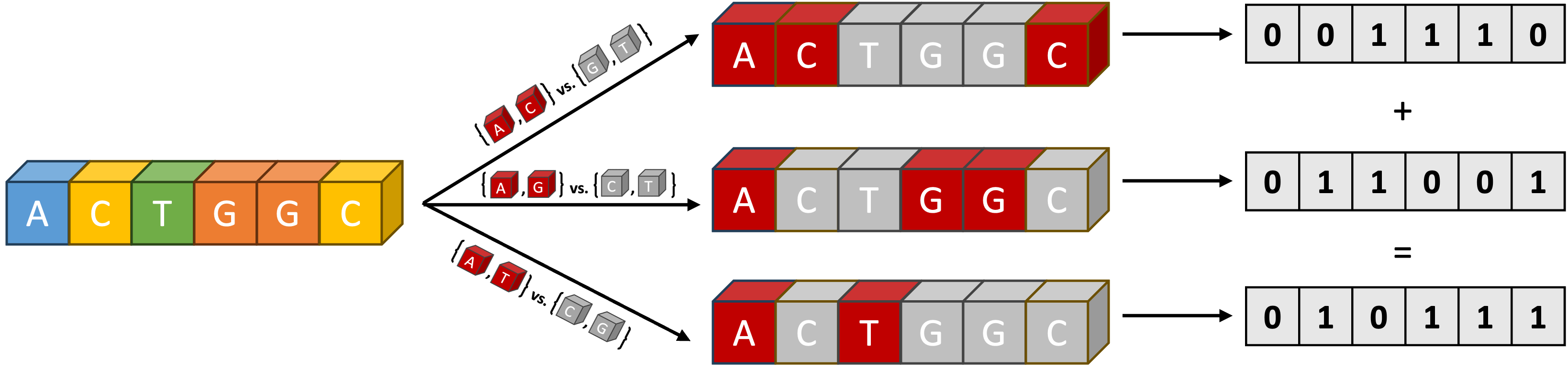}
    \caption{For the sequence `$AGGTC$', the vectors obtained from the first, second, and third partitions are $01110$, $00011$, and $01101$, respectively. It holds that $01110 + 00011 = 01101$.}
    \label{fig:ECC_seq}
    \vspace{-9pt}
\end{figure}
See Figure \ref{fig:ECC_seq} for an illustrative example.
%For example, for the sequence $AGGTC$, the vectors obtained from the first, second, and third partitions are $01110$, $00011$, and $01101$, respectively.
The resulting three (error-free) binary vectors satisfy an XOR-type constraint: one vector equals the bitwise sum (XOR) of the other two. This property is exemplified in Fig. \ref{fig:ECC_seq}. 
% As seen in Figure~\cite{fig:}
This corresponds precisely to the sum-channel structure with $\ell=2$, where two data rows induce a third parity row.
%Concretely, ECC sequencing uses the three distinct partitions:
%$\{A,C\}\ \text{vs.}\ \{G,T\},\ 
%\{A,G\}\ \text{vs.}\ \{C,T\},\ 
%\{A,T\}\ \text{vs.}\ \{C,G\},$
%and records a bit $0$ if the true base lies in the first class (and $1$ otherwise). For example, for the sequence $AGGTC$, the three binary vectors are $01110,\  00011,\  01101,$ and indeed $01110 + 00011 = 01101$.
Note that while two partitions suffice to reconstruct the original strand, a third partition adds redundancy that can be exploited for detecting and correcting sequencing errors.
In~\cite{chen2017highly}, decoding is posed as a probabilistic inference problem and solved by selecting the most likely sequence under a model validated on \emph{natural} genomes. Because synthetic DNA sequences are engineered and may violate these statistical assumptions, the resulting decoding performance may not transfer.
By contrast, in this work, we develop purely coding-theoretic error-correcting codes for the same model, making our framework directly applicable to \emph{synthetic} DNA data storage.

This paper is organized as follows.
Section~\ref{sec:preliminaries} introduces the notation used throughout the paper, defines the sum channel and its error model, and presents the formal formulation of the problem.
In Section~\ref{sec:1,2_D}, we construct a code that corrects two deletions for the sum-channel, for a two-row input matrix and then extend the construction to matrices with an arbitrary number of rows.
We then prove the order-optimality of the construction for a two-row matrix in Section \ref{sec:UB}.
In Section~\ref{sec:1,1_SID}, we present a simple, strictly optimal construction for two-row matrices that corrects a single edit over the sum channel. We then extend it to an explicit construction for binary matrices with any number of rows and show that its redundancy is within one bit of optimal.

\section{Preliminaries and Problem Formulation}
\label{sec:preliminaries}
 
Throughout this work, we denote $ [n] = \{1, \ldots, n\} $. Let $ \mathbb{F}_2 $ be the field of size two.
Denote a binary matrix as $ \X = (\x_1, \ldots, \x_\ell) \in \mathbb{F}_2^{\ell \times n} $, where $ \{\x_i\}_{i=1}^{\ell} $ are binary row vectors of length $n$, which means $ \x_i = (x_{i,1}, \ldots, x_{i,n}) \in \mathbb{F}_2^n $ for any $ i \in [\ell] $. For a vector $\x=(x_1,\ldots,x_n)\in\mathbb{F}_2^n$, we denote by
$\overline{\x} = (1+x_1,\ldots,1+x_n)\in\mathbb{F}_2^n$
its bitwise complement.
We further define ${parity}(\x) \triangleq \left( \sum_{i=1}^n x_i \right) \bmod 2$, 
and the \emph{derivative} of $\x$ by
$\Delta\x = (x_2+x_1,x_3+x_2,\ldots,x_n+x_{n-1}) \in \mathbb{F}_2^{n-1}.$ 
For indices $1\le i\leq  j\le n$, we denote by $\x[i:j]$ 
the contiguous subvector $(x_i,x_{i+1},\ldots,x_{j})$ of $\x$. We use the symbol $\circ$ to denote the \emph{concatenation} of binary vectors. For example, for $\bfu \in \mathbb{F}_2^{n_1}$ and $\bfv \in \mathbb{F}_2^{n_2}$, 
the concatenation $\bfu \circ \bfv \in \mathbb{F}_2^{n_1 + n_2}$ 
is defined as
$\bfu \circ \bfv = (u_1, \ldots, u_{n_1}, v_1, \ldots, v_{n_2}).$
For two binary vectors $\x,\y \in \F_2^n$, the
\emph{Hamming distance} between them is defined as
$\ham(\x,\y)
\triangleq
\left|\{\, i \in [n] : x_i \ne y_i \,\}\right|.$

\begin{definition}
    Let $\X= (\x_1, \ldots, \x_\ell) \in \mathbb{F}_2^{\ell\times n} $. The \textbf{ sum matrix} $ \X^{+} \in \mathbb{F}_2^{(\ell+1) \times n} $ of the matrix $\X$ is defined as the matrix whose first $ \ell $ rows are the rows of $\X$, and whose last row is their bitwise parity. Formally,
    \vspace{-1.1ex}
    $$ 
    \X^{+} = (
        \x_1 , \ 
        \x_2 , 
        \dots,
        \x_\ell ,\ 
        \boldsymbol{z}
    ), \quad \text{where } \boldsymbol{z} = \x_1  +  \x_2  +  \cdots  +  \x_\ell.
    $$
\end{definition}

\begin{example}
    The sum matrix of the matrix $\X = (1100,0000,1010)\in \F_2^{3\times 4}$ is $\X^+ = (1100,0000,1010, 0110)\in \F_2^{4\times 4}$.
\end{example}

Our channel model, referred to as the \emph{sum channel}, takes as input a binary matrix $\X \in \mathbb{F}_2^{\ell \times n}$ with $\ell<n$,
and outputs a binary matrix with $\ell+1$ rows. 
When no errors occur, the output is the \emph{sum matrix} $\X^+ \in \mathbb{F}_2^{(\ell+1)\times n}$ obtained from~$\X$ as defined above.

To formalize the error model, we begin by characterizing errors at the row level. 
The types of errors we consider are: substitutions ($\mathbb{S}$),
deletions ($\mathbb{D}$), insertions ($\mathbb{I}$),
as well as combinations thereof. 
For example, $\mathbb{ID}$ denotes a combination of insertions and deletions (indels), 
and $\mathbb{SID}$ represents substitutions, insertions, and deletions (edits). 
We use $\mathbb{T}$ as a placeholder for any of these error types. 

For a binary vector $\x \in \mathbb{F}_2^n$, let $\mathbb{T}_t(\x)$ denote the set of all vectors that can be obtained from $\x$ by introducing up to $t$ errors of type $\mathbb{T}$.

\begin{definition}
    Let $ \X \in \mathbb{F}_2^{\ell\times n} $, and let $ \X^{ + } \in \mathbb{F}_2^{(\ell+1) \times n} $ be its corresponding sum matrix, where $\ell,n \in \mathbb{N}$. For any integer $ t > 0 $, and an error type $ \mathbb{T} $, the \textbf{$t_{\mathbb{T}}$-error ball} for the sum channel centered at $ \X $, denoted $ B^{\mathbb{T}}_{t}(\X) $, is the set of all binary matrices that can be obtained by introducing up to $ t $ errors of type $ \mathbb{T}$ in $\X^+$. Formally,
\begin{equation*}
\begin{split}
        B^{\mathbb{T}}_{t}(\X)  = \big\{ 
\Y=(\y_1,\dots,\y_{\ell+1}) \ :\ \exists
s_1,\dots,s_{\ell+1}\in\mathbb{N},\\ \sum_{i=1}^{\ell+1} s_i\le t,\ 
\forall i\in[\ell+1],\ \y_i\in\mathbb{T}_{s_i}(\x_i)\ \big\},
\end{split}
\end{equation*}
where $\x_{\ell+1}$ is the bitwise sum of the upper rows.
\end{definition}
The following example illustrates the definitions of $ B^{\mathbb{T}}_{t}$.

\begin{example}
    Consider $\X= (11,10)$. We have $\X^ + = (11,10,01)$.
    The $1_{\mathbb{S}}$-error ball centered at $\X$ is given by $B^{\mathbb{S}}_{1}(\X) =
    \{
    (01,10,01),\ (10,10,01),\
    (11,00,01),\\ (11,11,01),\ 
    (11,10,11),\ (11,10,00) ,\ (11,10,01)
    \}$. 
    %The $(1,2)$-substitution ball centered at $\X$ is 
    %$B^{\mathbb{S}}_{(1,2)}(\X) =B^{\mathbb{S}}_{(1,1)}(\X)\cup 
    %\{   
    %(01,00,01),\ (10,00,01),\ (01,11,01),\\ (10,11,01),\ 
    %(01,10,11),\ (10,10,11),\ (01,10,00),\\ (10,10,00),\ 
    %(11,00,11),\ (11,00,00),\ (11,11,11),\\ (11,11,00)
    %\}.$
\end{example}

\begin{definition}
    For $\ell\in \mathbb{N},\ t>0$, and error type $\mathbb{T}$, a set $\mathcal{C}\subseteq \mathbb{F}_2^{\ell \times n}$ is called a \textbf{$(\ell;t)_\mathbb{T}$-correcting code} if for any $\X_1,\X_2\in \code$, we have 
    $B^{\mathbb{T}}_{t}(\X_1) \cap B^{\mathbb{T}}_{t}(\X_2) = \emptyset.$
     The maximum cardinality of a $(\ell;t)_\mathbb{T}$-correcting code is denoted by $A^\mathbb{T}_n(\ell;t)$, where the optimal redundancy is defined as $n\ell -\lceil\log_2 A^\mathbb{T}_n(\ell;t)\rceil$.
\end{definition}

We next show that the classical deletion--insertion equivalence~\cite{Levenshtein1966} extends to the sum-channel model: any code correcting $t$ deletions also corrects any combination of deletions and insertions of total size at most $t$.

\begin{claim}
Let $a,b\in\mathbb{N}$ with $a+b\le t$.
Any $(\ell;t)_{\mathbb{D}}$-correcting code for the sum channel can also correct any pattern of at most $a$ deletions and $b$ insertions. In particular, a code is $(\ell;t)_{\mathbb{D}}$-correcting if and only if it is $(\ell;t)_{\mathbb{I}}$-correcting.
\end{claim}

\begin{proof}
Assume $\mathcal C$ is $(\ell;t)_{\mathbb D}$-correcting, and suppose toward a contradiction that there exist distinct $\X,\Y\in\mathcal C$ and a matrix $\Z$ that can be obtained from both $\X^+$ and $\Y^+$ using at most $a$ deletions and $b$ insertions in total, where $a+b\le t$.

Write the $k$th rows of $\X^+,\Y^+,\Z$ as $\x_k,\y_k,\z_k$, respectively. For each $k\in[\ell+1]$, let $a_k,b_k$ be the numbers of deletions and insertions used to transform $\x_k$ into $\z_k$, and let $a_k',b_k'$ be the corresponding numbers for $\y_k$. Then
\[
\sum_{k=1}^{\ell+1}(a_k+b_k)\le a+b,\qquad
\sum_{k=1}^{\ell+1}(a_k'+b_k')\le a+b,
\]
and, since both $\x_k$ and $\y_k$ have length $n$ while the resulting row is $\z_k$,
\[
n-a_k+b_k = |\z_k| = n-a_k'+b_k',
\]
which implies
$a_k+b_k' = a_k'+b_k.$

Choose index sets $I_k,J_k\subseteq [|\z_k|]$ such that $|I_k|=n-a_k$, $|J_k|=n-a_k'$, $\z_k|_{I_k}$ is a subsequence of $\x_k$, and $\z_k|_{J_k}$ is a subsequence of $\y_k$. Then
\[
|I_k\cap J_k|
\ge |I_k|+|J_k|-|\z_k|
= n-(a_k+b_k').
\]
Hence $\x_k$ and $\y_k$ admit a common subsequence $\w_k$ of length at least
$n-t_k$, where $t_k\triangleq a_k+b_k'=a_k'+b_k$.
Let $\boldsymbol{W}$ be the matrix whose $k$th row is $\w_k$. Then $\boldsymbol{W}$ can be obtained from both $\X^+$ and $\Y^+$ by deleting at most $t_k$ symbols from row $k$, and
\[
\sum_{k=1}^{\ell+1} t_k
=
\sum_{k=1}^{\ell+1} a_k + \sum_{k=1}^{\ell+1} b_k'
\le a+b \le t.
\]
Therefore,
$\boldsymbol{W}\in B_t^{\mathbb D}(\X)\cap B_t^{\mathbb D}(\Y),$
contradicting the assumption that $\code$ is $(\ell;t)_{\mathbb D}$-correcting.
The converse direction is identical with the roles of insertions and deletions interchanged.
\end{proof}

Hence, it suffices to study deletion errors, since the corresponding insertion and mixed deletion-insertion cases follow immediately.

The sum row already introduces redundancy into the sum channel. The challenge is to exploit this redundancy to build error-correcting codes with as little additional redundancy as possible.

%Note that for all values of $t,\ell$ and $\mathbb{T}\in \{\mathbb{I},\mathbb{D}\}$, the largest $ (t,1,\ell)_{\mathbb{T}} $ correction code includes all matrices in $\mathbb{F}_2^{\ell \times n}$. This is because the erroneous row can be easily identified by comparing the lengths of the rows, and once identified,  it can be recovered by summing the last row of the corresponding sum matrix with all the remaining (non-erroneous) rows. Hence, $A_n^{\mathbb{I}}(t,1,\ell)=A_n^{\mathbb{D}}(t,1,\ell)= 2^{\ell n}$.

\section{$(\ell;2)_{\mathbb{D}}$-Correcting Codes}
\label{sec:1,2_D}
%Let $\ell\in\mathbb{N}$ and let $\X\in\mathbb{F}_2^{\ell\times n}$ be the channel input, with sum matrix $\X^{ + }$. We consider an error model in which the output matrix of the sum-channel is obtained from $\X^{ + }$ by applying at most one deletion in two rows. 
Our goal in this section is to construct 
%codes that enable perfect recovery under this model. That is, we aim to construct
$(\ell;2)_{\mathbb{D}}$-correcting codes, which enable the correction of two deletions in the output matrix of the sum channel.
%We note that all results in this section also apply to the corresponding insertion model.

First, suppose that exactly one row is erroneous. This case is easy to detect by comparing the row lengths. The original input can then be recovered by summing the parity row with all remaining (error-free) rows, yielding the missing row. Hence, the interesting case is when the deletions occur in different rows. We focus on this setting in this section and the next.

A natural approach is to encode each of the $\ell$ input rows independently using a Varshamov-Tenengolts (VT) code~\cite{varshamov1965code}. 
Whenever two rows are erroneous at the output of the sum channel, at least one of them must be among the first $\ell$ rows (which are VT codewords). Thus, we can recover that row via standard VT decoding and consequently reconstruct the original input matrix. 
However, this straightforward approach incurs redundancy of roughly $\ell\log_2 (n+1)$ bits. 
In contrast, we show that the same correction capability can be achieved with redundancy of only $2\lceil\log_2\log_2 n\rceil +\log_2 \ell+ O(1)$ bits.

We first treat the case $\ell=2$, where we present a concrete construction and then leverage this construction to obtain a general scheme for arbitrary $\ell\in\mathbb{N}$.

\subsection{Construction for $\ell=2$}
We will define a new class of codes $\mathcal{C}_{c_1,b_1,c_2,b_2}(n)$, where $0\leq c_1,c_2< 4\lceil\log_2 n\rceil+1,\ b_1,b_2\in \{0,1\}$ and prove that they are indeed $(2,2)_{\mathbb{D}}$-correcting codes. 
We first introduce the two building blocks of the construction.

\paragraph{The $SVT_{c,b}(n,P)$ Code} introduced in~\cite{schoeny2017codes}, is a variant of the
Varshamov-Tenengolts (VT) code~\cite{varshamov1965code} that can correct
a single deletion, provided that the location of the deleted bit is known
to within $P$ consecutive positions.
\vspace{-1.1ex}
\begin{equation*}
\begin{split}
    SVT_{c,b}(n,P)\! =\! \{ \x\hspace{-0.25ex}\in\hspace{-0.25ex} \F_2^n   \hspace{-0.25ex}:\hspace{-0.25ex} \sum_{i=1}^n i  x_i \equiv c \bmod \hspace{-0.25ex} P ,\hspace{-0.25ex} \  parity(\x)\!=\!b \}.
\end{split}
\end{equation*}
As shown in~\cite{schoeny2017codes}, there exist $0 \le c < P$ and
$b \in \F_2$ such that the redundancy of the $SVT_{c,b}(n,P)$ code
is at most $\log_2 P + 1$ bits.

\paragraph{Misalignment Constraint Set}
%$\mathcal{P}^+_\ell(n,k)$ 
For parameters $n,k\in\mathbb{N}$ where $k<n$, define the set $\mathcal{L}(n,k)$, to be the set of pairs $(\bfa,\bfb)\in\F_2^{2\times n}$ that satisfy the following misalignment property: for any two windows $\w_{\bfa},\w_{\bfb}$ of length $k\!+\!1$ whose starting positions differ by at most one, we have
$\w_{\bfa} \neq \w_{\bfb}$ and  $\w_{\bfa} \neq \overline{\w}_{\bfb}$.
That is, we allow two equal length windows in $\bfa,\bfb$ whose starting positions differ by at most one to be equal or complement if and only if their length is at most $k$.
In our constructions we enforce this misalignment constraint by comparing \emph{derivatives} rather than the windows themselves: if $\Delta \w_{\bfa} \neq \Delta \w_{\bfb}$ (where derivatives have length $k$), then necessarily $\w_{\bfa}$ and $\w_{\bfb}$ are neither identical nor complements.
Formally, the $\mathcal{L}(n,k)$ is defined as follows:
    \begin{equation*}
\begin{split}
    \mathcal{L}(n,k) = \Big\{& ( \bfa, \bfb  )\! \in \!\mathbb{F}_2^{2\times n} : \ \forall i\in\!  [n\! -\! k],\ \delta\in\! \{ -1,0,1\}, \\ & 
    \Delta\bfa[i:i\!+\!k\!-\!1] \neq \Delta\bfb[i\!+\!\delta:i\!+\!k\!-\!1+\delta]
    \Big\}.
\end{split}
\end{equation*}

\begin{comment}
    \begin{equation*}
\begin{split}
    \mathcal{L}(n,k) = \Big\{& ( \bfa, \bfb  ) \in \mathbb{F}_2^{2\times n} : \ \forall i\in  [n- k],\\ & 
\Delta\bfa[i:i+k-1] \notin \big\{ \Delta\bfb[i-1:i+k-2],\\ & \Delta\bfb[i:i+k-1],\ \Delta\bfb[i+1:i+k] \big\} 
\Big\}.
\end{split}
\end{equation*}
\end{comment}

\begin{example}
    Consider the vectors $\bfa = 1110110,\   \bfb = 1010010\in \F_2^7$.
    Their derivatives are $\Delta \bfa = 001101$ and $\Delta \bfb = 111011$.
    Every window of length $5$ in $\Delta \bfa$ does not appear in $\Delta \bfb$
    either in the same starting position or with a shift of one index. Hence,
    $(\bfa,\bfb) \in \mathcal{L}(7,5)$. On the other hand,
    $(\bfa,\bfb) \notin \mathcal{L}(7,4)$, since 
    $\Delta\bfa[3:6] = 1101 = \Delta\bfb[2:5]$. 
\end{example}
\begin{comment}
    \begin{example}
    Consider the vectors $\bfa = 0101110$ and $\bfb = 1110100$ in $\F_2^7$.
    Their derivatives are $\Delta \bfa = 111001$ and $\Delta \bfb = 001110$.
    Every window of length $3$ in $\Delta \bfa$ does not appear in $\Delta \bfb$
    either in the same starting position or with a shift of one index. Hence,
    $(\bfa,\bfb) \in \mathcal{L}(7,3)$. On the other hand,
    $(\bfa,\bfb) \notin \mathcal{L}(7,2)$, since 
    $\Delta\bfa[2:3] = 11 = \Delta\bfb[3:4]$. 
\end{example}
\end{comment}

Using the definition of $\mathcal{L}(n,k)$, we introduce the global misalignment constraint set $\mathcal{P}^+_\ell(n,k)$, for $\ell \in \mathbb{N}$.
The set $\mathcal{P}^+_\ell(n,k)$ consists of all binary matrices with $\ell$ rows and $n$ columns such that any two rows of their sum matrix satisfy the misalignment property defined by $\mathcal{L}(n,k)$.
Formally,
\begin{equation*}
\begin{split}
    \mathcal{P}^+_\ell (n,k)\ =\ &\{ \X = (\x_1,\dots \x_\ell) \in \F_2^{\ell \times n} \ : \\ &
\forall i_1, i_2 \in [\ell+1],\ i_1\neq i_2,
        (  \x_{i_1},  \x_{i_2}) \in \mathcal{L}(n,k) \},
\end{split}
\end{equation*}
where $\x_{\ell+1} = \x_1 + \dots + \x_\ell$.

%Intuitively, the constraints reduce alignment ambiguity, allowing correction of two deletions when their column positions are more than $k$ apart.

% We next state a lemma used in Algorithm~\ref{decoding_alg}, showing that the input constraint set prevents ambiguity for two-bit insertions in two positions of distance larger than $k$: the two insertion possibilities do not yield the same bitwise-sum output.

%This global misalignment constraint reduces ambiguity when the deletions are far apart: for deletion positions at distance larger than $k$, there is no alternative way to place the two deletions that yields the same received parity row, and hence the same observed output.

This global misalignment constraint reduces ambiguity when the two deletions are far apart. In particular, it rules out inputs for which two distinct correction patterns, with deletion positions at distance greater than $k$, yield the same parity row. Lemma~\ref{sum_amb} formalizes this.

\begin{lemma}\label{sum_amb}
    %Let $\bfa,\bfb\in \F_2^n$ and let $\w_{\bfa}, \w_{\bfb}$ be two windows of length $k\!+\!1$ in $\bfa,\bfb$, respectively, whose starting positions differ by at most one. If for some $u,v \in\F_2$ we have ....
    %then, $(\bfa,\bfb)\not\in \mathcal{L}(n,k)$.
    Let $\w_1,\w_2\in\F_2^{k+1}$ and $v_1,v_2\in\F_2$ satisfy
    \vspace{-1.1ex}
    $$
    (v_1\circ \w_1)+(\w_2\circ v_2)=(\w_1\circ v_1)+(v_2\circ \w_2).
    \vspace{-1.1ex}
    $$
    Consider any $\bfa,\bfb\in\F_2^n$ such that $\w_1$ occurs in $\bfa$ and $\w_2$ occurs in $\bfb$, with starting indices that differ by at most one. Then $(\bfa,\bfb)\notin\mathcal{L}(n,k)$.
\end{lemma}
\begin{proof}
    Assume $\w_1,\w_2\in\F_2^{k+1}$ satisfy the stated condition. We have $\forall i\in[k]$, 
    $\w_1[i] +  \w_2[i+1]= \w_1[i+1] +  \w_2[i] .
    $ Then, $
    \w_1[i] +  \w_1[i+1] = \w_2[i] +  \w_2[i+1],$
    which implies that $\forall i\in[k], \quad(\Delta \w_1)[i]= (\Delta \w_2)[i], $ therefore, $
    \Delta \w_1[1:k]= \Delta \w_2[1:k].$ 
    By the definition of $\mathcal{L}(n,k)$, if $\w_1$ appears in $\bfa$ and $\w_2$ appears in $\bfb$ with starting indices differing by at most one, then $(\bfa,\bfb)\notin\mathcal{L}(n,k)$.

\end{proof}

\begin{example}
    Consider the two windows $\w_1 = 0110,\ \w_2 = 1001$ in the vectors $\bfa = 1110110= 111\circ \w_1,\   \bfb = 1010010= 10\circ \w_2 \circ0 $. We have that the starting positions of the two windows $\w_1,\w_2$ differ by one index.
    Consider the bits $v_1= 1,\ v_2= 0$.
    We have that $(v_1 \circ \w_1)  +  (\w_2\circ v_2)  = 10110\: + \: 10010 = 00100$.
    Also, $(\w_1\circ v_1)  +  (v_2 \circ\w_2) = 01101\: +\: 01001 = 00100.$ 
    Note that $\Delta\w_1 = 101 = \Delta \w_2$. This implies that $\bfa,\bfb\not\in \mathcal{L}(7,3)$.
\end{example}

We now want to calculate the size of the global misalignment constraint set. We first calculate the size of the set $\mathcal{L}(n,k)$.

\begin{lemma}\label{L_size}
    For $n,k\in \mathbb{N}$, $|\mathcal{L}(n,k)|\ge 2^{2n}\cdot \big(1-\frac{3(n-k)}{2^k}\big).$
\end{lemma}
\begin{proof}
    Fix $\bfa \in \F_2^n$.
    We will establish an upper bound for the number of vectors $\bfb$ such that $(\bfa,\bfb)$ does not belong to $\mathcal{L}(n,k)$.

    We will calculate the probability that there exists an index $i\in [n-k]$ such that the window $\Delta\bfa[i:i\!+\!k\!-\!1]$ appears in a uniformly random vector $\Delta\bfb$, at one of the aligned positions $i-1,i,i+1$. 
    Each such match occurs with probability $2^{-k}$. By the union bound over the three aligned positions, the probability of a match for a single window is at most $3\cdot 2^{-k}$.
    Since there are $n-k$ possible window positions in $\Delta\bfa$, using union bound, the probability that
    a collision occurs in some window is at most $3(n-k)/2^k$. Then
    the probability that no collision occurs is at least $1-3(n-k)/2^k$.
\end{proof}

\begin{comment}
\begin{lemma} \label{P_l-size}
    For $n,k,\ell\! \in\! \mathbb{N}$,
    $| \mathcal{P}_\ell^+ (n,k)| \ge 2^{n\ell} \big(1-\frac{ 3(n-k)}{2^k}\big)^{\binom{\ell+1}{2}} $. 
    $$2^{n\ell}\left(1-\binom{\ell+1}{2}\frac{3(n-k)}{2^k}\right)$$
\end{lemma}
\begin{proof}
    The probability that a pair $(\bfa,\bfb)\in \mathcal{L}(n,k)$ is $p={|\mathcal{L}(n,k)|}/{2^{2n}}$. By the definition of $\mathcal{P}^+_\ell$ there are $\binom{\ell+1}{2}$ independent pairwise constraints. Hence, the probability that a matrix $\X \in \mathcal{P}_\ell^+ (n,k)$ is bounded from below by $p^{\binom{\ell+1}{2}}$. By Lemma~\ref{L_size}, we have $p= 1-\frac{3(n-k)}{2^k}.$
\end{proof}
\end{comment}

\begin{lemma}\label{P_l-size}
For $n,k,\ell\!\in\!\mathbb{N}$,
$|\mathcal{P}_\ell^+\!(n,k)|
\ge
2^{n\ell}\left(1\!-\!\binom{\ell+1}{2}\frac{3(n-k)}{2^k}\right).$
\end{lemma}
\begin{proof}
Let $\X$ be chosen uniformly at random from $\F_2^{\ell\times n}$, and let
$\x_{\ell+1}=\sum_{i=1}^{\ell}\x_i$ denote its parity row.
For each unordered pair $\{s,t\}\subseteq[\ell+1]$, define the bad event
$E_{s,t}\triangleq \{(\x_s,\x_t)\notin\mathcal{L}(n,k)\}.$
Then $\X\in\mathcal{P}_\ell^+(n,k)$ if and only if none of the events
$\{E_{s,t}\}_{\{s,t\}\subseteq[\ell+1]}$ occurs.

For every pair $\{s,t\}\subseteq[\ell+1]$, the pair $(\x_s,\x_t)$ is uniformly distributed over $\F_2^n\times\F_2^n$.
Indeed, this is immediate when $s,t\in[\ell]$, and if, say, $t=\ell+1$, then
$\x_{\ell+1}=\x_s+\sum_{i\in[\ell]\setminus\{s\}}\x_i,
$
where $\sum_{i\in[\ell]\setminus\{s\}}\x_i$ is uniform over $\F_2^n$ and independent of $\x_s$.
Hence, by Lemma~\ref{L_size},
\[
\Pr[E_{s,t}]
=
1-\frac{|\mathcal{L}(n,k)|}{2^{2n}}
\le
\frac{3(n-k)}{2^k}.
\]

Since there are $\binom{\ell+1}{2}$ such pairs, the union bound gives
\[
\Pr\bigl[\X\notin\mathcal{P}_\ell^+(n,k)\bigr]
\le
\binom{\ell+1}{2}\frac{3(n-k)}{2^k}.
\]
Therefore,
\[
\Pr\bigl[\X\in\mathcal{P}_\ell^+(n,k)\bigr]
\ge
1-\binom{\ell+1}{2}\frac{3(n-k)}{2^k},
\]
and multiplying by the total number $2^{n\ell}$ of matrices in $\F_2^{\ell\times n}$ yields the claim.
\end{proof}

The following corollary is obtained by evaluating the bound in Lemma~\ref{P_l-size} at $k=4\lceil\log_2 n\rceil$ and using the fact that $\ell<n$.
\begin{corollary}\label{cor:P_l}
    $| \mathcal{P}_\ell^+ (n,4\lceil\log_2 n\rceil)|\ge 2^{n\ell}\left(1-\frac{3}{n}\right)$.
\end{corollary}

With these definitions, we now specify the code construction $\mathcal{C}_{c_1,b_1,c_2,b_2}(n)$ and prove its $(2;2)_{\mathbb{D}}$-correcting property. 

\begin{construction}\label{first_const}
    For $0\leq c_1,c_2< 4\lceil\log_2 n\rceil {+1},\ b_1,b_2\in \{0,1\}$, we define:
    \begin{equation*}
    \begin{split}
        \mathcal{C}_{c_1,b_1,c_2,b_2}(n)
        = \Big\{&
            \X = (\x_1, \x_2) \in \mathbb{F}_2^{2 \times n} :\\ &
             \forall i \in [2],\ \x_i \in SVT_{c_i,b_i}(n, 4\lceil\log_2 n\rceil {+1})
             ,\\ & 
            \X \in \mathcal{P}^{+}_2( n, 4\lceil\log_2 n\rceil)
          \Big\}.
    \end{split}
    \end{equation*}
    That is, $\mathcal{C}_{c_1,b_1,c_2,b_2}(n)$ consists of all matrices $\X$ where $i$-th row belongs to an $SVT_{c_i,b_i}(n, 4\lceil\log_2 n\rceil\!+\!1)$ code, and the sum matrix $\X^{+}$ satisfies the global misalignment constraint. % specified by $\mathcal{L}$.
\end{construction}

\begin{theorem}\label{main_thm} 
The code
$\mathcal{C}_{c_1,b_1,c_2,b_2}(n)$ is a $(2;2)_{\mathbb{D}}$-correcting code.
\end{theorem}
To prove the theorem, we show that Algorithm~\ref{decoding_alg} correctly decodes any received output of the sum channel corrupted by two deletions.
The algorithm distinguishes between two cases according to the distance between the affected columns: either the two deletions occur in columns that are far apart, or the two affected columns lie within an interval of length $4\log n+1$. In the first case, the global misalignment constraint guarantees that there is a unique way to correct the received matrix so that the sum relation is restored. In the second case, there may be more than one possible correction. However, since one deletion in a row is known to lie within an interval of length $4\lceil\log n\rceil+1$, the corresponding $SVT$ decoder can be used to identify the correct one.
%, for any fixed choice of parameters $0 \leq c_1,c_2< \lceil\log_2 n\rceil{  +c+1}$ and $b_1,b_2\in \{0,1\}$.

%Before detailing the algorithm, we note that the exact position of a deleted bit within a run is irrelevant. Hence, deletions are treated at the run level rather than by explicit bit indices.
%
% Also note that in a correctly aligned configuration, every position in the output matrix $\Y=(\y_1,\y_2,\y_3)$ satisfies $\y_1[i] \! +\!  \y_2[i] \!= \!\y_3[i]$. A single deletion in one of the rows misaligns the vectors, causing the sum relation to fail in the subsequent run. Similarly, when scanning from right to left, the same reasoning applies in reverse. 

\begin{alg}\label{decoding_alg} 
    Let $\Y=(\y_1,\y_2, \y_3)$, be the output matrix of the sum-channel. Perform the following steps:
    \begin{enumerate}[label={}, leftmargin=*]
    
    \item \textbf{Step 1: Rows detection.}
    Compare the lengths of the three rows and identify the shorter ones, whose indices are denoted by $i_1,i_2$. Denote the error-free row index by $i_3$.
    If exactly one row is shorter, recover it by summing the other two rows, and return the corrected matrix.

    \item \textbf{Step 2: Columns detection.}
    %Detect the runs $r_1,r_2$ in which the deletions occur by 
    Examine the columns of $\Y$ in two scans:
    \begin{enumerate}
        \item
        Sequentially check each triplet
        $(\y_1[j],\y_2[j],\y_3[j])$ for $j=1,\ldots,n-1$, and record the first index $j_1$ at which the sum relation fails.
        If no such index exists, set $j_1=n$.
        
        \item
        Scan in the opposite direction, checking each triplet
        $(\y_{i_1}[j-1],\y_{i_2}[j-1],\y_{i_3}[j])$ for $j=n,\ldots,2$, and record the first index $j_2$ at which the sum relation fails.
        If no such index exists, set $j_2=1$.
    \end{enumerate}
     %Thus, $\widehat{j}_1$ is the first position of the $r_1$-th run from the left, and $\widehat{j}_2$ is the first position of the $r_2$-th run from the right, with $r_1$ to the left of $r_2$.
     
     \item \textbf{Step 3: Case analysis.} 
     %To determine which affected position (run) corresponds to which row, proceed as follows:
     %Having identified the erroneous rows and runs, the next step is to determine the correspondence between each erroneous rows and its run.
    % Based on the relative positions of $\widehat{j}_1$ and $\widehat{j}_2$, proceed as follows:
    \begin{itemize}
        \item If $|j_2-j_1| \leq 4\lceil \log_2 n\rceil$,
        set $i=\min\{i_1,i_2\}\in[2]$. Use the $SVT_{c_i,b_i}(n,4\lceil \log_2 n\rceil +1)$ decoder to recover the deletion in row $i$.
        Then sum the corrected row $i$ and the error-free row to obtain the remaining corrected row, and return the reconstructed matrix.
        
        %then both deletions, in particular the one in the first row, occurred within the interval $[\widehat{j}_1,\widehat{j}_2]$ whose length is at most $\lceil \log_2 n \rceil{  +c+1}$. 
    
        \item 
        Otherwise,
        determine the deleted bits $v_1,\ v_2$ in rows $i_1,\ i_2$, respectively, using the known row parities $b_1,b_2$ and the fact that the parity row has parity $b_1+b_2$.
        %Lemma~\ref{sum_amb} and the definition of $\mathcal{P}_2^{+}(n,\lceil\log_2 n\rceil\!+\!c)$ imply a 
        Find the unique insertion of $v_1,v_2$ at $j_1,j_2$ such that
        the resulting matrix satisfies the sum relation.
        %$\y_1[j_1\!:\!j_2]+\y_2[j_1\!:\!j_2] =\y_3[j_1\!:\!j_2]$.
        Insert $v_1,v_2$ accordingly and return that matrix.
    \end{itemize}
\end{enumerate}
\end{alg}

We will now prove that Algorithm~\ref{decoding_alg} can correct up to two deletions.

\begin{proof}
    Note that a row affected by a deletion is shorter than the other rows in the matrix. Hence, Step~1 correctly identifies the erroneous rows.

    Recall that in a correctly aligned configuration, every position in the output matrix $\Y=(\y_1,\y_2,\y_3)$ satisfies
    $\y_1[i]+\y_2[i]=\y_3[i]$.
    A single deletion in one of the rows does not affect the sum relation within the run containing the deleted bit, since deleting any symbol from a run yields the same resulting subsequence. However, it shifts all subsequent symbols in that row by one position, causing the sum relation to fail at the first index after the affected run. We therefore associate the deletion with that index. Similarly, when scanning from right to left, the same argument applies in reverse. Hence, Step~2 identifies the positions of the deletions.

    First, consider the case where the deletion positions are at most $4\lceil\log_2 n\rceil$ bits apart, that is,
    $|j_1-j_2|\leq 4\lceil\log_2 n\rceil.$
    Since there are only three rows and $i_1\neq i_2$, it follows that
    $i=\min\{i_1,i_2\}\in[2]$.
    Therefore, we can apply the corresponding $SVT$ decoder to row $i$, and then reconstruct the original matrix as described.

    Now consider the case where the two deletions are farther apart, namely,
    $
    |j_1-j_2| > 4\lceil\log_2 n\rceil.
    $
    Then, by Lemma~\ref{sum_amb} and the fact that any pair of rows of the original sum matrix belongs to $\mathcal{L}(n,4\lceil\log_2 n\rceil)$, in particular the pair corresponding to rows $i_1$ and $i_2$, there is a unique way to insert the bits $v_1,v_2$ at positions $j_1,j_2$ in rows $i_1,i_2$ such that $\y_1[j_1\!:\!j_2]+\y_2[j_1\!:\!j_2]=\y_3[j_1\!:\!j_2].$
    
    Thus, in both cases, Algorithm~\ref{decoding_alg} successfully recovers the corrupted matrix.
\end{proof}

We state the redundancy of our construction in the following corollary.
\begin{corollary}
    For sufficiently large $n$, there exists a code $\mathcal{C}_{c_1,b_1,c_2,b_2}\hspace{-0.25ex}(n)$ with at most 
    $2\lceil\log_2\hspace{-0.25ex}\log_2 \hspace{-0.15ex}n \rceil\hspace{-0.25ex}+\hspace{-0.15ex}O(1)$ redundancy bits.
\end{corollary}
\begin{proof}
By Corollary~\ref{cor:P_l},
$|\mathcal{P}_2^+(n,4\lceil\log_2 n\rceil)| \ge 2^{2n}\left(1-\frac{3}{n}\right).$
Since the sets $SVT_{c,b}(n,4\lceil\log_2 n\rceil+1)$ partition $\F_2^n$ as $(c,b)$ ranges over
$0\le c<4\lceil\log_2 n\rceil+1$ and $b\in\{0,1\}$, the codes
$\mathcal{C}_{c_1,b_1,c_2,b_2}(n)$ partition $\mathcal{P}_2^+(n,4\lceil\log_2 n\rceil)$ into
$4(4\lceil\log_2 n\rceil+1)^2$ disjoint sets. Hence, by the pigeonhole principle, there exist
$c_1,c_2,b_1,b_2$ such that
\[
|\mathcal{C}_{c_1,b_1,c_2,b_2}(n)|
\ge
\frac{2^{2n}\left(1-\frac{3}{n}\right)}{4(4\lceil\log_2 n\rceil+1)^2}.
\]
For sufficiently large $n$, this is at least
$\frac{2^{2n}}{100\lceil\log_2 n\rceil^2}.$
Therefore, the redundancy is at most
$2n-\log_2|\mathcal{C}_{c_1,b_1,c_2,b_2}(n)|
\le
\log_2(100)+2\lceil\log_2\log_2 n\rceil$.
\end{proof}

\begin{comment}
    \begin{proof}
    By Corollary~\ref{cor:P_l} there are at least $2^{2n}\big(1 - \frac{3}{n}\big)$ matrices in $\F_2^{2\times n}$ that satisfy the global misalignment constraint.
    The $4(4\lceil\log_2 n\rceil +1)^2$ options of the values $c_1,c_2,b_1,b_2$ partition the space $F_2^{2\times n}$ into $4(4\lceil\log_2 n\rceil {+1})^2$ disjoint sets. Thus, using the pigeonhole principle, there exists a choice of $c_1,c_2,b_1,b_2$ for which the cardinality of the corresponding code is at least $$2^{2n}\cdot\frac{1-3/n}{4\cdot(4\lceil\log_2 n\rceil {+1})^2}
    \ge
    \frac{2^{2n}}{100\cdot\lceil\log_2 n\rceil^2}
    $$ 
    \end{proof}
\end{comment}

\subsection{Construction for Arbitrary $\ell$}

We present a construction based on the tensor product code framework introduced in~\cite{wolf2006introduction}.

\begin{comment}
\paragraph{SVT Variant} Let $n,P \in \mathbb{N}$. Denote $h= \lceil\log_2 P\rceil$.
%, such that $\F_{2^h}$, the field of $2^h$ elements, is the smallest extension field of $\mathbb{F}_2$, that has at least $P$ elements.
%
For $0\leq c < 2^h, \ b\in \F_2$ define
\begin{equation*}
\begin{split}
    {SVT}^{2^h}_{c,b}(n,P) = \{ \bfz=(z_1,\dots z_n) \in \mathbb{F}_2^n :  \  
    parity(\bfz) = b,\\
    \sum_{j=1}^n j\cdot  z_j \equiv c \bmod {2^h}  \ \}.
\end{split}
\end{equation*}
\ey{why do we need $2^h$ in the notation of the code?}
One can verify that ${SVT}^{2^h}_{c,b}(n,P)$ can correct a single deletion within $P$ consecutive positions using the same decoding algorithm as the original SVT codes.
\end{comment}

For $q\in\mathbb{N}$ and $\x\in\F_2^n$, define the $q$-syndrome of $\x$ by
$Syn_q(\x)\triangleq \sum_{j=1}^n j x_j \pmod q.$
Let
$P\triangleq 4\lceil \log_2 n\rceil +1
$ and $
h\triangleq \lceil \log_2 P\rceil.$
Fix an arbitrary bijection
$\phi:\mathbb Z_{2^h}\times\{0,1\}\to \F_{2^{h+1}}.$
Then, for every $\x\in\F_2^n$, define
$\sigma(\x)\triangleq \phi\big(Syn_{2^h}(\x),parity(\x)\big)\in\F_{2^{h+1}}.$

\begin{construction}\label{TP_ID}
    Let $\mathcal{C}_{2}$ be a length-$\ell$ code over $\mathbb{F}_{2^{h+1}}$ of minimum Hamming distance 3, i.e., it can correct two erasures. Then, define the code
\begin{equation*}
\begin{split}
    \mathcal{C}_{\text{TP}}(\ell,n, \mathcal{C}_{2}) = \Big\{ \X=(\x_1,&\dots,\x_{\ell})\in \mathbb{F}_2^{\ell\times n} \ : \\ &
    (\sigma(\x_1), \sigma(\x_2),\dots,\sigma(\x_\ell))\in \mathcal{C}_{2}, \\ &
    \X \in \mathcal{P}_\ell^+ (n,4\lceil \log_2 n \rceil)
    \Big\}.
\end{split}
\end{equation*}
\end{construction}
As noted earlier, Construction~\ref{TP_ID} produces a tensor-product code based on the $2^h$-syndromes and the parity of the input matrix rows.
Next, we prove that the code $\mathcal{C}_{\text{TP}}(\ell,n, \mathcal{C}_{2})$ is a $(\ell;2)_{\mathbb{D}}$-correcting code.

\begin{theorem}
The code
     $\mathcal{C}_{\text{TP}}(\ell,n, \mathcal{C}_2)$ is a $(\ell;2)_{\mathbb{D}}$-correcting code.
\end{theorem}
\begin{proof}
    Denote by $i_1,i_2$ the indices of the shorter (erroneous) rows. Since $\ell-2$ rows are error free, it is possible to compute the values $\sigma(\x_j)$ for $j\in [\ell]\setminus \{ i_1,i_2\}$.
    Since $\mathcal{C}_2$ corrects two erasures, we can recover $\sigma(\x_{i_1})$ and $\sigma(\x_{i_2})$. Consequently, we can retrieve the $2^h$-syndromes and the row parities of the erroneous rows, which suffice to reconstruct $\x_{i_1}$ and $\x_{i_2}$ using a similar decoding procedure as in Algorithm~\ref{decoding_alg}.
\end{proof}

\begin{comment}
\begin{theorem}\label{redundancy_thm}
Let $ R^\star $ be the number of redundancy symbols of $ \mathcal{C}_{\mathbb{D}_2}^\star$, a linear code of length $\ell$ that corrects 2 deletions over $ \mathbb{F}_{2^{h+1}} $. Then, there exists a code $ \mathcal{C}_{\text{TP}}(\ell,n, \mathcal{C}_{\mathbb{D}_2}^\star) $ with redundancy at most  
$
 R^\star \cdot (h+1) - 9\cdot 2^{  c}\cdot\ell^2
$ bits.

\end{theorem}
\begin{proof}
Note that any coset of $\mathcal{C}_{\mathbb{D}_2}^\star$ is a distinct 2 deletions correcting code, and there are $q^{R^\star}$ cosets, where $ q = 2^{h+1} $.
Thus, there are $ q^{R^\star} $ distinct $ \mathcal{C}_{\text{TP}}(\ell,n, \cdot) $ codes, one for each coset, and all these distinct codes form a partition of the space of all the binary matrices in $\mathcal{P}_\ell^+(n,P-1)$.
Using the pigeonhole principle and Lemma~\ref{P_l-size}, we get that there exists a code with at least 
$     ({ |\mathcal{P}_\ell^+(n,P-1)|})/{q^{R^\star}} \ge 
2^{\ell n - R^\star \cdot (h+1) - 9\cdot 2^{  c}\,\ell^2 } $ codewords.
\end{proof}
\end{comment}

We will now state the redundancy of our construction.

%
%
%The claim follows by the pigeonhole principle.

\begin{corollary}
    For $\ell<n$, there exists a $\mathcal{C}_{\mathrm{TP}}(\ell,n,\mathcal{C}_2)$ code with redundancy at most
    $2\left\lceil \log_2 \log_2 n \right\rceil +\log_2 \ell +  O(1)$
    bits.
\end{corollary}
\begin{proof}
    By Corollary~\ref{cor:P_l}, there are at least
    $2^{n\ell}\left(1-\frac{3}{n}\right)$
    matrices in $\F_2^{\ell\times n}$ that satisfy the global misalignment constraint.

    Let
    $q \triangleq 2^{h+1},\ 
    r \triangleq \left\lceil \log_q \ell \right\rceil +1.$
    Consider a shortened $q$-ary Hamming code $\mathcal{C}_2$ of length $\ell$,
    redundancy $r$, and minimum Hamming distance at least $3$.
    Hence, $\mathcal{C}_2$ can correct two erasures and has size
    $|\mathcal{C}_2| = q^{\ell-r}.$
    Therefore, $\F_q^\ell$ is partitioned into $q^r$ cosets of $\mathcal{C}_2$.

    By the pigeonhole principle, there exists a coset containing at least
    $2^{n\ell}\frac{1-\frac{3}{n}}{q^r}$
    matrices satisfying the global misalignment constraint. Hence, there exists
    a choice of $\mathcal{C}_2$ in Construction~2 such that the resulting code
    $\mathcal{C}_{\mathrm{TP}}(\ell,n,\mathcal{C}_2)$ has redundancy at most
    \[
    r\log_2 q-\log_2\left(1-\frac{3}{n}\right).
    \]
    Since $q=2^{h+1}$, we obtain
    \[
    r\log_2 q = r(h+1)
    \le
    \left(\log_q \ell +2\right)(h+1)
    =
    \log_2 \ell +2(h+1).
    \]
    By the definition of $h$,
    \[
    2(h+1)=2\left\lceil \log_2\log_2 n\right\rceil + O(1).
    \]
    Also,
    $-\log_2\left(1-\frac{3}{n}\right)=O(1).$
    Thus, the redundancy is at most
    $\log_2 \ell + 2\left\lceil \log_2\log_2 n \right\rceil + O(1),$
    as claimed.
\end{proof}

\section{Upper Bound on the Size of $(2;2)_{\mathbb{D}}$-Correcting Codes}
\label{sec:UB}

In this section we show that the redundancy of any $(2;2)_{\mathbb{D}}$-correcting code is at least $\lceil\log_2\log_2 n \rceil+ O(1)$. 
Therefore, we consider the \emph{confusability graph} $\mathcal{G}(n)$ whose vertex set is $\F_2^{2\times n}$ and every two vertices $\X_1,\X_2\in\F_2^{2\times n}$ are adjacent if and only if $B^{\mathbb{D}}_{2}(\X_1) \cap B^{\mathbb{D}}_{2}(\X_2)\neq \emptyset$.
Obviously, any $(2;2)_{\mathbb{D}}$-correcting code is an independent set in the graph $\mathcal{G}(n)$. To provide an upper bound on the size of an independent set in $\mathcal{G}(n)$, we use the concept of clique covers.

\begin{definition}
        A collection $\mathcal{Q}$ of cliques is a clique cover of $\mathcal{G}$ if every vertex in $\mathcal{G}$ belongs to some clique in $\mathcal{Q}$.
\end{definition}

\begin{lemma}[\hspace{0.05ex}\cite{knuth1994sandwich}]\label{indp_set}
If $\mathcal{Q}$ is a clique cover of $\mathcal{G}$, then the size of any independent set of $\mathcal{G}$ is at most $|\mathcal{Q}|$.
\end{lemma}

\begin{comment}
    Let $\mathcal{G}(\mathcal{V},\mathcal{E})$ be a graph whose set of vertices is all the binary matrices of size $2\times n$, i.e., $\mathcal{V}=\F^{2\times n}$.
For two distinct vertices $\X_1,\X_2\in \mathcal{V}$, we place an edge between them if and only if
$
B^{\mathbb{S}}_{(1,1)}(\X_1)\cap B^{\mathbb{S}}_{(1,1)}(\X_2)\neq\emptyset.
$
Equivalently, there is an edge between $\X_1$ and $\X_2$ if and only if $\ham(\X_1,\X_2) = 2$.
%
Obviously, any $(1,1,2)_{\mathbb{S}}$-correcting code is an independent set. 
\end{comment}
By Lemma~\ref{indp_set}, the size of any clique cover in $\mathcal{G}(n)$ is an upper bound on the size of 
$A^\mathbb{D}_n(2;2)$. We now give a clique cover construction.
% TODO: use redundancy

\paragraph{Clique Cover Construction}
Consider a parameter $k\in \mathbb{N}$. For simplicity, assume $n$ is a multiple of $k$ and set $m= n/k$. We divide each sequence of length $n$ to $m$ blocks where each block is of length $k$. Let $\bfe_j^{(0)}$ be the $j$-th unit vector of length $k$, and define $\bfe_j^{(1)}$ as its bitwise complement.
We define the following set of pairs of blocks of length $k$:
$$ \Lambda_k = \Big\{  (\bfe^{(b_1)}_j, \bfe^{(b_2)}_j) \in \F_2^{2\times k}  \ :\  b_1,b_2\in \{0,1\},\  j\in [k] \Big\}. $$
We also define $\widetilde{\Lambda}_k = \F_2^{2\times k} \setminus \Lambda_k$.
It holds that $ |\Lambda_k| = 4k$, and $|\widetilde{\Lambda}_k| = 2^{2k}-4k$.

The clique cover consists of singletons and cliques of size $k$.
The singletons are of the form
$ S_{\X}\! =\! \{ \X\!\in\! \F_2^{2\times n} : (\x_1, \x_2)\!\in\! (\widetilde{\Lambda}_k)^m \}$,
whereas the $k$-cliques are constructed based on the following set
\begin{equation*}
\begin{split}
    \Gamma_k = \{ (\bfu_1,\bfu_2,\w_1,\w_2,i) :&(\bfu_1,\bfu_2)\in \widetilde{\Lambda}_k^{i-1},\\  &\w_1,\w_2\in \F_2^{k(m-i)}, \ i\in [m] \} .
\end{split}
\end{equation*}
For each $\bfz = (\bfu_1,\bfu_2,\w_1,\w_2,i)\in \Gamma_k$, we define 4 sets of vertices, which we will later show to be cliques of size $k$: 
\begin{equation*}
\begin{split}
    Q_{\bfz}^{(b_1,b_2)} =\Big\{ \big( \bfu_1\circ \bfe^{(b_1)}_j \circ \w_1\ , \ \bfu_2 \circ \bfe^{(b_2)}_j \circ \w_2 \big) \ : \ j\!\in\! [k] \Big\},
\end{split}
\end{equation*}
where $b_1,b_2\in \{0,1\}$. Finally, we define:
$$ \mathcal{Q}(n,k)\! =\! 
\Big\{ S_{\X} : (\x_1, \x_2)\!\in \!(\widetilde{\Lambda}_k)^m \Big\}
\cup
\Big\{ 
Q_{\bfz}^{(b_1,b_2)} \!
\Big\}_{\bfz\in \Gamma_k }^{ b_1,b_2\in \{0,1\}}. $$

\begin{comment}
\begin{claim}\label{3rd_row}
    Fix $\bfz \in \Gamma_k$ and $b_1, b_2 \in \{0,1\}$.
    Then, all matrices in $Q_{\bfz}^{(b_1,b_2)} \subseteq \mathbb{F}_2^{2\times n}$ have the same parity row.
\end{claim}
\begin{proof}
    Let $\X=
    \big( \bfu_1 \circ\bfe^{(b_1)}_j \circ\w_1\ , \ \bfu_2 \circ\bfe^{(b_2)}_{j} \circ\w_2 \big),\
    \X' =
    \big( \bfu_1 \circ\bfe^{(b_1)}_{j'}  \circ\w_1\ , \ \bfu_2 \circ\bfe^{(b_2)}_{j'} \circ\w_2 \big) 
    \in Q_{\bfz}^{(b_1,b_2)}$ where $\bfz = (\bfu_1,\bfu_2,\w_1,\w_2,i)$.
    It holds that,
    \begin{equation*}
    \begin{split}
        \bfu_1 \bfe^{(b_1)}_j \w_1 +  \bfu_2 \bfe^{(b_2)}_{j} \w_2 
        &=
        (\bfu_1\! +\!  \bfu_2)\circ (b_1\!+\! b_2)^\ell \circ (\w_1\! +\!  \w_2) \\
        &= 
        \bfu_1 \bfe^{(b_1)}_{j'} \w_1 + \bfu_2 \bfe^{(b_2)}_{j'} \w_2.
    \end{split}
    \end{equation*}
\end{proof}
\end{comment}

\begin{lemma}
    For $n,k\in \mathbb{N}$, the set $\mathcal{Q}(n,k)$ is a clique cover.
\end{lemma}
\begin{proof}
    To prove that $\mathcal{Q}(n,k)$ is a clique cover, we need to show that every set in $\mathcal{Q}(n,k)$ is a clique and all vertices in these sets cover the whole space $\F_2^{2\times n}$. 

    Clearly, each singleton is a clique.
    Fix $\bfz = (\bfu_1,\bfu_2,\w_1,\w_2,i)\in \Gamma,\  b_1,b_2 \in \F_2$. We will show $Q_{\bfz}^{(b_1,b_2)}$ forms a clique. Let $\X = \big( \bfu_1 \circ\bfe^{(b_1)}_j \circ\w_1\ , \ \bfu_2 \circ\bfe^{(b_2)}_{j}\circ \w_2 \big),\
    \X' = \big( \bfu_1 \circ\bfe^{(b_1)}_{j'}\circ \w_1\ , \ \bfu_2\circ \bfe^{(b_2)}_{j'} \circ\w_2 \big)$ be two matrices in $Q_{\bfz}^{(b_1,b_2)}$.
    Consider $\X^+, \X'^+\in \F_2^{3\times n}$, the sum matrices of $\X, \X'$, respectively. 
    %By Claim~\ref{3rd_row}, the third row of those two matrices is identical.
    Note that the third row of these two matrices is identical since the parity row is determined solely by the surrounding blocks $u_1,u_2,w_1,w_2$.
    By deleting the $j$-th symbol of $\bfe^{(b_1)}_j,\ \bfe^{(b_2)}_j$ in the first, second row of $\X^+$, respectively, we obtain the matrix
    $ \big( \bfu_1\circ b_1^{\ell-1} \circ\w_1\ , \ \bfu_2\circ b_2^{\ell-1} \circ\w_2 \ ,\ \x_3\big) ,$
    which is identical to the matrix we obtain by deleting the $j'$-th symbol of $\bfe^{(b_1)}_{j'},\ \bfe^{(b_2)}_{j'}$, from the first, second row of $\X'^+$.
    Hence, $Q_{\bfz}^{(b_1,b_2)}$ is a clique.

    It remains to show that any $\X\in \F_2^{2\times n}$ belongs to some clique in $\mathcal{Q}(n,k)$. If $\X\in \widetilde{\Lambda}_k^m $, then $\X\in S_{\X}$.
    Otherwise, $\X\not\in \widetilde{\Lambda}_k^m$ and one of the $m$ subblocks of $\X$ belongs to $\Lambda_k$. Let the $i$-th subblock be the first subblock from the left that belongs to $\Lambda_k$. Hence, this subblock is of the form $(\bfe^{(b_1)}_j, \bfe^{(b_2)}_j) $ for some $  b_1,b_2\in \{0,1\},\  j\in[k]$. This implies that $\X$ belongs to $Q_{(\bfu_1,\bfu_2,\w_1,\w_2,i)}^{(b_1,b_2)}$, where $\bfu_1,\bfu_2$ is the first $i-1$ subblocks of $\x_1,\x_2$, respectively, and $\w_1,\w_2$ is the last $m-i$ subblocks of $\x_1,\x_2$, respectively.
\end{proof}

\paragraph{Clique Cover Size Analysis}
We will calculate the size of the clique cover and then prove that any $(2;2)_\mathbb{D}$-correcting code has redundancy of at least $\lceil\log_2\log_2 n \rceil+O(1)$.

\begin{lemma}\label{lem:clique_cover_size}
    The size of $\mathcal{Q}(n,k)$ is given by
    $$ 4^{n}\cdot \Big\{ \Big( 1-\frac{4k}{4^{k}} \Big)^{n/k} 
    +
    \frac{1}{k}\Big( 1- \Big( 1-\frac{4 k}{4^{k}} \Big)^{n/k} \Big) \Big\} .
    $$
    %$$\leq 4^{n}\cdot \Bigg\{ \Big( 1-\frac{4k}{4^{k}} \Big)^{n/k} + \frac{1}{k} \Bigg\}.$$
\end{lemma}
\begin{proof}
    By construction, the number of singletons is $|\widetilde{\Lambda}_k|^m$, while the number of $k$-cliques is 
    $$ 4|\Gamma|  = 4 \sum_{i=1}^m |\widetilde{\Lambda}_k|^{i-1}  \cdot 2^{2k(m-i)}  = 2^{2+2km-2k} \sum_{i=1}^m \Big(\frac{|\widetilde{\Lambda}_k|}{2^{2k}}\Big)^{i-1}.$$
    Using the formula for the sum of geometric series we have that $$\sum_{i=1}^m \Big(\frac{|\widetilde{\Lambda}_k|}{2^{2k}}\Big)^{i-1} =
    \frac{1-\Big(\frac{|\widetilde{\Lambda}_k|}{2^{2k}}\Big)^m}{1-\frac{|\widetilde{\Lambda}_k|}{2^{2k}}}.$$
    Recall that $|\widetilde{\Lambda}_k|=2^{2k}-4k$ and $n=m\cdot k$. Thus, the number of $k$-cliques is
    $$ 2^{2+2km-2k} \cdot \frac{1-\big(1-\frac{4k}{2^{2k}}\big)^m }{\frac{4k}{2^{2k}}} = \frac{2^{2n}}{k} \Big( 1- \Big(1-\frac{4k}{2^{2k}}\Big)^{n/k} \Big),$$
    and the size of $\mathcal{Q}(n,k)$ is
    $$ 2^{2n}\Big( 1-\frac{4k}{4^{k}} \Big)^{n/k} 
    +
    \frac{2^{2n}}{k}\Big( 1- \Big( 1-\frac{4 k}{4^{k}} \Big)^{n/k} \Big). $$
\end{proof}

\begin{corollary}
    Let $\mathcal{C}\subseteq \F_2^{2\times n}$ be a $(2;2)_{\mathbb{D}}$-correcting code. The redundancy of $\mathcal{C}$ is at least $ \lceil\log_2\log_2 n \rceil +O(1)$.
\end{corollary}
\begin{proof}
    Define $f(n) \triangleq \left(1-\frac{4k}{4^k}\right)^{n/k}.$
    Then $\mathcal{Q}(n,k)= 4^n\left(f(n)+\frac{1}{k}\cdot\left(1-f(n)\right)\right).$
    Since $f(n) \in (0,1)$ for any $n,k>1$, then
    $\mathcal{Q}(n,k)\leq 4^n\left(f(n)+\frac{1}{k}\right) $. Using the inequality $(1-x)\le e^{-x}$ for $x\in(0,1)$, we obtain
    $$f(n) \le \exp\!\left(-\frac{4k}{4^k}\cdot \frac{n}{k}\right) = \exp\!\left(-\frac{4n}{4^k}\right).$$

    Set $k=\left\lceil \frac{1}{6}\log_2 n\right\rceil$. Then $2^{2k} \ge n^{1/3},$
    and hence
    $f(n)
    \le \exp\!\left(-4n^{2/3}\right).$
    In particular, for sufficiently large $n$, $\exp(-4n^{2/3}) \le 1/k$, and therefore
    $$
    \mathcal{Q}(n,k)\le 4^n\left(\frac{1}{k}+\exp(-4n^{2/3})\right)
    \le 4^n\cdot \frac{2}{k}.
    $$

Finally, since $k=\left\lceil \frac{1}{6}\log_2 n\right\rceil \ge \frac{1}{6}\log_2 n$, we get
$$A_n^{\mathbb{D}}(2;2)\leq  \mathcal{Q}(n,k)\le 4^n\cdot \frac{12}{\log_2 n}.$$
\end{proof}

Recall that Construction~\ref{first_const} achieves redundancy $2\lceil \log_2\log_2 n\rceil\! +\! O(1)$, which is within a factor of $2$ of the lower bound.

\section{$(\ell;1)_{\mathbb{SID}}$-Correcting Code}
\label{sec:1,1_SID}
Let $\ell\! \in\! \mathbb{N}$. Consider the case where for an input $\X\!\in\! \mathbb{F}_2^{\ell \times n}$ to the sum channel,
the output matrix $\Y$ is obtained by introducing
one edit error in one of the rows of $\X^ + $. 
We will present a code construction that enables the correction of a single edit.  
When the error is an insertion or a deletion, the original input can be reconstructed directly, as previously explained. Hence, the nontrivial case to consider is when the error is a substitution.  

We begin with $\ell=2$, describing a $(2;1)_{\mathbb{SID}}$-correcting code and proving optimality. We then extend the construction to all $\ell\in\mathbb{N}$, obtaining redundancy within one bit of optimal.

\subsection{Construction and Optimality for $\ell=2$}

A simple construction is obtained by fixing the parities of the two rows of the input matrix, yielding an $(2;1)_{\mathbb{SID}}$-correcting code with only two redundancy bits. 
Indeed, a single substitution in $\X^{+}$ flips exactly two parities: the parity of the corrupted row and the parity of the corrupted column. 
By definition of the sum matrix, every column of $\X^{+}$ has even parity. 
Consequently, the error location is determined by the unique row and column that violate the parity constraint, and the error is corrected by flipping the bit at their intersection. One can prove that this construction achieves \emph{strictly} optimal redundancy
This construction is formalized in Construction~\ref{simple_construction}.

\begin{construction}\label{simple_construction}
    Let $n\in\mathbb{N}$. Fix $b_1,b_2\in \mathbb{F}_2$. 
    \begin{equation*}
    \begin{split}
        \mathcal{C}_{b_1,b_2}(n) = 
        \Big\{ & 
            \X = (\x_1,\x_2)\in \mathbb{F}_2^{2 \times n} : \\ & \forall i\in[2],\  parity(\x_i)=b_i
        \Big\}.
    \end{split}
    \end{equation*}
\end{construction}

The following theorem can be obtained as previously explained.
\begin{theorem}
     The code $\mathcal{C}_{b_1,b_2}(n)$ is an $(2;1)_{\mathbb{SID}}$-correcting code of size $2^{2n-2}$.
\end{theorem}

To upper bound the size of any $(2;1)_{\mathbb{SID}}$-correcting code, we consider the \emph{confusability graph} $\mathcal{G}(n)$ whose vertex set is $\F_2^{2\times n}$.
Two vertices $\X_1,\X_2\in\F_2^{2\times n}$ are adjacent if and only if
$B^{\mathbb{SID}}_{1}(\X_1)\cap B^{\mathbb{SID}}_{1}(\X_2)\neq\emptyset$.
Equivalently, $\X_1$ and $\X_2$ are adjacent if and only if $\ham(\X_1,\X_2)\le 2$.
We then use a clique-cover argument to obtain an upper bound on the code size.

\begin{lemma}\label{clique_const}
    For any $\x_1,\x_2\in\F^{n-1}$, define the set of matrices
    \begin{equation*}
    \begin{split}
        Q(\x_1,\x_2) =& \{
    (\x_1\circ 0 , \x_2\circ 0),\ 
    (\x_1\circ 1 , \x_2\circ 0),\\  &
    (\x_1\circ 0 , \x_2\circ 1),\
    (\x_1\circ 1 , \x_2\circ 1)
    \} \subseteq \F^{2\times n}.
    \end{split}
    \end{equation*}
    Then the set $ \mathcal{Q} = \{ Q(\x_1,\x_2)\ :\ \x_1,\x_2\in\F_2^{n-1} \} $ forms a clique cover of size $2^{2n-2}$.
\end{lemma}
\begin{proof}
    Let $\x_1,\x_2\in\F^{n-1}$.
    We first prove that the set $Q(\x_1,\x_2)$ forms a clique.
    For any two distinct matrices $\boldsymbol{A},\boldsymbol{B}\in Q(\x_1,\x_2)$, their sum matrices $\boldsymbol{A}^{ + }$ and $\boldsymbol{B}^{ + }$ differ in exactly two entries, both in the last column. Equivalently, $\ham(\X_1,\X_2) = 2$.
    This implies that there is an edge between every pair of distinct matrices in $Q(\x_1,\x_2)$; therefore, the subgraph induced by $Q(\x_1,\x_2)$ is a clique.

     It can be easily seen that each matrix belongs to a set $\{ Q(\x_1,\x_2)\}_{\x_1,\x_2\in\F^{n-1}}$, hence, $\mathcal{Q}$ is a clique cover of size $2^{2n-2}$.
\end{proof}

\begin{corollary}
    For any $b_1,b_2\in \mathbb{F}_2$, $\mathcal{C}_{b_1,b_2}(n)$ is optimal, i.e., $|\mathcal{C}_{b_1,b_2}(n)|=A^\mathbb{SID}_n(2;1)$.
\end{corollary}
\begin{proof}    
    In Lemma~\ref{clique_const} we construct a clique cover of size $2^{2n-2}$. By Lemma~\ref{indp_set}, it follows that any independent set , i.e. $(2;1)_{\mathbb{S}}$-correcting code, is of size  at most $2^{2n-2}$.

    Since every $(2;1)_{\mathbb{SID}}$-correcting code is also an $(2;1)_{\mathbb{S}}$-correcting code, then $A^\mathbb{SID}_n(2;1)\leq A^\mathbb{S}_n(2;1)$ and the theorem is proven.
\end{proof}

\subsection{Construction and Upper Bound for Arbitrary $\ell$}

A natural extension of the $\ell=2$ construction to an arbitrary~$\ell$ can be obtained by fixing the parity of each of the $\ell$ rows. This approach incurs a redundancy of $\ell$ bits. 
Now, we present a construction that achieves significantly improved redundancy. 
We again draw on ideas from tensor-product codes, as presented in~\cite{wolf2006introduction}.

%\subsection{Construction and Optimality for $\ell=2$}

\begin{construction}\label{TP_sub}
     Let $\ell,n\in\mathbb{N}$ and let $\mathcal{C}_{\mathbb{S}_1}$ be a binary single-substitution-correcting code of length $\ell$. Then, the following code $\mathcal{C}_{TP}(\ell,n,\mathcal{C}_{\mathbb{S}_1})$ is defined to be
     %an $(1,1,\ell)_{\mathbb{SID}}$-correcting code,
    $$ \mathcal{C}_{TP}(\ell,n,\mathcal{C}_{\mathbb{S}_1}) = 
    \Big\{
    \X\in \F_2^{\ell\times n} \ :\ (b_1,b_2, \ldots , b_\ell)\in\mathcal{C}_{\mathbb{S}_1} 
    \Big\},
    $$
    where $b_i =parity(\x_i)$.
\end{construction}

\begin{comment}
    \begin{construction}\label{construction1}
    Let $n\in\mathbb{N}$. Fix $b_1,b_2\in \mathbb{F}_2$. 
    \begin{equation*}
    \begin{split}
        \mathcal{C}_{b_1,b_2}(n) = 
        \Big\{ & 
            \X = (\x_1,\x_2)\in \mathbb{F}_2^{2 \times n} : \\ & \forall i\in[2],\  parity(\x_i)=b_i
        \Big\}.
    \end{split}
    \end{equation*}
\end{construction}
\end{comment}

\begin{theorem}
    $ \mathcal{C}_{TP}(\ell,n,\mathcal{C}_{\mathbb{S}_1})$ is a $(\ell;1)_{\mathbb{SID}}$-correcting code.
\end{theorem}
\begin{proof}
We present a decoding algorithm for the code $\mathcal{C}_{TP}(\ell,n,\mathcal{C}_{\mathbb{S}_1})$, showing it is an $(\ell;1)_{\mathbb{SID}}$-correcting code.
Let $\X = (\x_1,\dots,\x_\ell)\in \mathcal{C}_{TP}(\ell,n,\mathcal{C}_{\mathbb{S}_1})$ be the transmitted matrix whose sum matrix is $\X^+$ and let $\Y = (\y_1,\dots,\y_{\ell+1})\in \F_2^{(\ell+1)\times n}$ be the sum-channel output that differs from $\X^{ + }$ by a single substitution.

Note that a single-bit substitution flips exactly two parities: that of the corrupted row and that of the corrupted column.
From the construction of the sum matrix, we have that the parity of each column in $\X^{ + }$ is 0.
We therefore locate the unique column with non-zero parity and denote it by $\widehat{j}$.
Then we compute the row parities of $\Y$ to obtain ${\bfb'}=({b}'_1,\ldots,{b}'_\ell)$ with ${b}'_i=parity(\y_{i})$. It holds that $\ham\big((b_1,\ldots, b_\ell), ({b}'_1,\ldots,{b}'_\ell)\big) \leq 1$. 
Therefore, using the decoder of $\mathcal{C}_{\mathbb{S}_1}$, 
we either identify an erroneous row $\widehat{i}\in[\ell]$ or conclude that no row among the first $\ell$ is corrupted (which occurs when the substitution is in the last row).
If an erroneous row $\widehat{i}\in[\ell]$ is found, flip the entry $(\widehat{i},\widehat{j})$ of $\Y$ and output its first $\ell$ rows.
\end{proof}

\begin{comment}

\begin{lemma}\label{code2-size}
    For any $\ell,n\in\mathbb{N}$ and a single-substitution-correcting code $\mathcal{C}_{\mathbb{S}_1}$ over $\F_2^\ell$,
    $| \mathcal{C}_{TP}(\ell,n,\mathcal{C}_{\mathbb{S}_1}) | =
    |\mathcal{C}_{\mathbb{S}_1}|\cdot 2^{\ell(n-1)}.
    $
\end{lemma}
\begin{proof}
    Each $\bfb\in\mathcal{C}_{\mathbb{S}_1}$ fixes the parity of every row of matrices in $\mathcal{C}_{TP}$. 
    Each row with a prescribed parity has $2^{n-1}$ possibilities, independently of the others. 
    Hence the number of matrices corresponding to a given vector $\bfb$ is $(2^{n-1})^{\ell}$, and therefore $\lvert \mathcal{C}_{TP} \rvert \;=\; \lvert \mathcal{C}_{\mathbb{S}_1} \rvert \cdot 2^{\ell(n-1)}.$
\end{proof}
\end{comment}

Next, we state the redundancy of our construction.
\begin{corollary}
    There exists a $(\ell;1)_{\mathbb{SID}}$-correcting code with redundancy of at most $\lceil \log_2 (\ell+1)\rceil$ bits.
\end{corollary}
\begin{proof}
    Let $\mathcal{C}_{\mathbb{S}_1}\subseteq \F_2^\ell$ be a single-substitution-correcting code. Then each $\bfb\in\mathcal{C}_{\mathbb{S}_1}$ fixes the parity of every row of matrices in $\mathcal{C}_{TP}$. 
    Each row with a prescribed parity has $2^{n-1}$ possibilities, independently of the others. 
    Hence the number of matrices corresponding to a given vector $\bfb$ is $(2^{n-1})^{\ell}$, and therefore $\lvert \mathcal{C}_{TP} \rvert \;=\; \lvert \mathcal{C}_{\mathbb{S}_1} \rvert \cdot 2^{\ell(n-1)}.$

    Recall that a length-$\ell$ shortened Hamming code $\mathcal{C}_{Ham}$ with minimum distance $d=3$ can correct a single substitution, and is of size $2^{\ell- \lceil\log_2(\ell+1)\rceil}$. 
    Thus, $|\mathcal{C}_{TP}(\ell,n,\mathcal{C}_{Ham})|=2^{n\ell- \lceil\log_2(\ell+1)\rceil}$.
\end{proof}

Lastly, we derive an upper bound on the size of any $(\ell;1)_{\mathbb{SID}}$-correcting code, showing that the redundancy of our construction is within a single bit of optimality for any $\ell$. We work with the $1$-substitution \emph{sphere} around $\X$, i.e., the set of outputs obtained by exactly one substitution, that is, $B^{\mathbb{S}}_{1}(\X)\setminus \{\X\}$.

\begin{claim}\label{1S-sphere-size}
    For any $\X \in \F_2^{\ell\times n}$, $|B^{\mathbb{S}}_1(\X)\setminus\{\X\} | = (\ell+1)\cdot n$.
\end{claim}
\begin{proof}
    A single substitution in the sum matrix of $\X$, $\X^+\in \F_2^{(\ell+1)\times n}$, is specified by a choice of a row in $[\ell+1]$ and a choice of a column in $[n]$. Thus, the number of distinct erroneous matrices obtainable from $\mathbf{X}^{+}$ is $(\ell+1)n$. 
\end{proof} 

\begin{claim}\label{output_space_size}
    $|\bigcup_{\mathbf{X}\in\F_2^{\ell\times n}} \big(B^{\mathbb{S}}_{1}(\X)\setminus \{\X\}\big)| = 2^{n \ell}\cdot n$.
\end{claim}
\begin{proof}
    We now compute the size of the output space of the sum channel in the case of exactly one substitution. 
    Formally, we evaluate
    $|\bigcup_{\mathbf{X}\in\mathbb{F}_2^{\ell\times n}} B^{\mathbb{S}}_{1}(\mathbf{X})|.$
    The output space consists of all matrices in $\mathbb{F}_2^{(\ell+1)\times n}$ for which exactly one column has odd parity and all other columns have even parity. 
    There are $n$ choices for the column with odd parity, and for any fixed choice, the remaining degrees of freedom yield $2^{\ell n}$ matrices. 
    Therefore, the total number of such matrices is $n\cdot 2^{\ell n}$.
\end{proof}

\begin{theorem}
    $A^\mathbb{S}_n(\ell;1) \leq \lfloor \frac{2^{n\ell}}{\ell+1}\rfloor$. 
\end{theorem}
\begin{proof}
    Let $\code$ be a $(\ell;1)_{\mathbb{S}}$-correcting code.
    By a sphere-packing bound argument~\cite{Hamming1950}, the size of any $(\ell;1)_{\mathbb{S}}$-correcting code is upper bounded by the ratio between the output space size and the size of an $1$-substitution sphere. By Claim~\ref{1S-sphere-size}, the size of such sphere is $(\ell+1)n$.
    Moreover, from Claim \ref{output_space_size}, the size of the output space is $n\cdot 2^{\ell n}$.
    Hence, $| \code |\leq \lfloor\frac{n\cdot 2^{\ell n}}{(\ell+1)n}\rfloor = \lfloor \frac{2^{n\ell}}{\ell+1}\rfloor$.
\end{proof}

Since any $(\ell;1)_{\mathbb{SID}}$-correcting code is, in particular, a $(\ell;1)_{\mathbb{S}}$-correcting code, we have $A^\mathbb{SID}_n(\ell;1)\le A^\mathbb{S}_n(\ell;1).$
Recall that Construction~\ref{TP_sub} for an $(\ell;1)_{\mathbb{SID}}$-correcting code achieves a code of size 
$2^{n\ell-\lceil \log_2(\ell+1)\rceil}$.
Since
$\ell+1 \le 2^{\lceil \log_2(\ell+1)\rceil} < 2(\ell+1),$
the size of our construction is within a multiplicative factor of at most $2$ of the sphere-packing upper bound, which corresponds to a gap of at most one redundancy bit.

%This proves the optimality of Construction \ref{TP_sub}. \ey{strictly optimal?}

\bibliographystyle{IEEEtran}
\bibliography{references}

@inproceedings{wolf2006introduction,
  title={An introduction to tensor product codes and applications to digital storage systems},
  author={Wolf, Jack Keil},
  booktitle={2006 IEEE Information Theory Workshop-ITW'06 Chengdu},
  pages={6--10},
  year={2006},
  organization={IEEE}
}

@article{schoeny2017codes,
  title={Codes correcting a burst of deletions or insertions},
  author={Schoeny, Clayton and Wachter-Zeh, Antonia and Gabrys, Ryan and Yaakobi, Eitan},
  journal={IEEE Transactions on Information Theory},
  volume={63},
  number={4},
  pages={1971--1985},
  year={2017},
  publisher={IEEE}
}

@article{knuth1994sandwich,
  author    = {Donald E. Knuth},
  title     = {The Sandwich Theorem},
  journal   = {The Electronic Journal of Combinatorics},
  volume    = {1},
  number    = {1},
  year      = {1994},
  note      = {Article 1},
}

@article{varshamov1965code,
  title={A code for correcting a single asymmetric error},
  author={Varshamov, R. R. and Tenengolts, G.},
  journal={Automatika i Telemekhanika},
  volume={26},
  number={2},
  pages={288--292},
  year={1965}
}

@article{chen2017highly,
  title={Highly accurate fluorogenic DNA sequencing with information theory--based error correction},
  author={Chen, Zitian and Zhou, Wenxiong and Qiao, Shuo and Kang, Li and Duan, Haifeng and Xie, X Sunney and Huang, Yanyi},
  journal={Nature biotechnology},
  volume={35},
  number={12},
  pages={1170--1178},
  year={2017},
  publisher={Nature Publishing Group US New York}
}

@article{Hamming1950,
  author  = {Hamming, R. W.},
  title   = {Error Detecting and Error Correcting Codes},
  journal = {Bell System Technical Journal},
  volume  = {29},
  number  = {2},
  pages   = {147--160},
  year    = {1950},
  doi     = {10.1002/j.1538-7305.1950.tb00463.x}
}

@inproceedings{raid,
  title={A case for redundant arrays of inexpensive disks (RAID)},
  author={Patterson, David A and Gibson, Garth and Katz, Randy H},
  booktitle={Proceedings of the 1988 ACM SIGMOD international conference on Management of data},
  pages={109--116},
  year={1988}
}

@article{chen1994raid,
  title={RAID: High-performance, reliable secondary storage},
  author={Chen, Peter M and Lee, Edward K and Gibson, Garth A and Katz, Randy H and Patterson, David A},
  journal={ACM Computing Surveys (CSUR)},
  volume={26},
  number={2},
  pages={145--185},
  year={1994},
  publisher={ACM New York, NY, USA}
}

@article{cai2019optimal,
  title={Optimal codes correcting a single indel/edit for DNA-based data storage},
  author={Cai, Kui and Chee, Yeow Meng and Gabrys, Ryan and Kiah, Han Mao and Nguyen, Tuan Thanh},
  journal={arXiv preprint arXiv:1910.06501},
  year={2019}
}

@article{blawat2016forward,
  title={Forward error correction for DNA data storage},
  author={Blawat, Meinolf and Gaedke, Klaus and Huetter, Ingo and Chen, Xiao-Ming and Turczyk, Brian and Inverso, Samuel and Pruitt, Benjamin W and Church, George M},
  journal={Procedia Computer Science},
  volume={80},
  pages={1011--1022},
  year={2016},
  publisher={Elsevier}
}

@article{Levenshtein1966,
  title={Binary coors capable or ‘correcting deletions, insertions, and reversals},
  author={Lcvenshtcin, VI},
  booktitle={Soviet physics-doklady},
  volume={10},
  number={8},
  year={1966}
}

\end{document}